\documentclass{article}
\usepackage{sarah}
\usepackage{comment}
\usepackage{xspace}
\usepackage[numbers, sort]{natbib}

\title{Incremental Voronoi diagrams}

\author{Sarah R.~Allen\thanks{Carnegie Mellon University.  Supported by NSF grants CCF-0747250, CCF-1116594, and the Graduate Research Fellowship Program under Grant No. DGE-1252522. \texttt{srallen@cs.cmu.edu}} \and Luis Barba\thanks{D\'{e}partment d'Informatique, Universit\'{e} Libre de Bruxelles, Carleton University, \texttt{lbarbafl@ulb.ac.be+}} \and John Iacono\thanks{Tandon School of Engineering, New York University, \texttt{iacono@nyu.edu}} \and Stefan Langerman\thanks{Directeur de Recherches du F.R.S.-FNRS, \texttt{stefan.langerman@ulb.ac.be}}}

\newcommand{\V}{\ensuremath{\mathcal V(S)}}
\newcommand{\Vq}{\ensuremath{\mathcal V_q(S)}}
\newcommand{\rr}[2]{\ensuremath{\textsc{cell}(#1,#2)}}
\newcommand{\rrb}[2]{\ensuremath{\partial \textsc{cell}(#1,#2)}}
\newcommand{\f}{\ensuremath{\mathcal F}}
\newcommand{\curve}{\ensuremath{\calC}}
\newcommand{\ins}[1][\calC]{\ensuremath{\textsc{in}(#1)}}
\newcommand{\out}[1][\calC]{\ensuremath{\textsc{ex}(#1)}}
\newcommand{\fleeq}[1][\curve]{\ensuremath{\calE_{#1}}}
\newcommand{\pres}[1][G, \calC]{\ensuremath{\calP(#1)}}
\newcommand{\flarb}[1][G, \fleeq]{\ensuremath{\calF(#1)}}
\newcommand{\cost}[1][G, \calC]{\ensuremath{\textsc{cost}(#1)}}
\newcommand{\bounded}[1][G, \curve]{\ensuremath{\mathcal{B}(#1)}}
\newcommand{\VD}{\textrm{VD}}
\newcommand{\calV}{\mathcal{V}}

\begin{document}

\newcommand{\shortfull}[2]{#2}

\maketitle

\begin{abstract}
We study the amortized number of combinatorial changes (edge insertions and removals) needed to update the graph structure of the Voronoi diagram $\VD(S)$ (and several variants thereof) of a set $S$ of $n$ sites in the plane as sites are added to the set. 
To that effect, we define a general update operation for planar graphs that can be used to model the incremental construction of several variants of Voronoi diagrams as well as the incremental construction of an intersection of halfspaces in $\mathbb{R}^3$. 
We show that the amortized number of edge insertions and removals needed to add a new site to the Voronoi diagram is $O(\sqrt{n})$. 
A matching $\Omega(\sqrt{n})$ combinatorial lower bound is shown, even in the case where the graph representing the Voronoi diagram is a tree. 
This contrasts with the $O(\log{n})$ upper bound of Aronov et al.~(2006) for farthest-point Voronoi diagrams in the special case where the points are inserted in clockwise order along their convex hull.

We then present a semi-dynamic data structure that maintains the Voronoi diagram of a set $S$ of $n$ sites in convex position.
This data structure supports the insertion of a new site $p$ (and hence the addition of its Voronoi cell) and finds the asymptotically minimal number $K$ of edge insertions and removals needed to obtain the diagram of $S \cup \{p\}$ from the diagram of $S$, in time $O(K\,\mathrm{polylog}\ n)$ worst case, which is $O(\sqrt{n}\;\mathrm{polylog}\ n)$ amortized by the aforementioned combinatorial result.

The most distinctive feature of this data structure is that the graph of the Voronoi diagram is maintained explicitly at all times and can be retrieved and traversed in the natural way; this contrasts with other known data structures supporting nearest neighbor queries.
Our data structure supports general search operations on the current Voronoi diagram, which can, for example, be used to perform point location queries in the cells of the current Voronoi diagram in $O(\log n)$ time, or to determine whether two given sites are neighbors in the Delaunay triangulation. 
\end{abstract}

\section{Introduction}

Let $S$ be a set of $n$ sites in the plane.  The graph structures of the Voronoi diagram $\VD(S)$ and its dual the Delaunay triangulation $\textrm{DT}(S)$ capture much of the proximity information of that set. They contain the nearest neighbor graph, the minimum spanning tree, and the Gabriel graph of $S$, and have countless applications in computational geometry, shape reconstruction, computational biology, and machine learning.

One of the most popular algorithms for constructing a Voronoi diagram inserts sites in random order, incrementally updating the diagram~\cite{de2000computational}. In that case, backward analysis shows that the expected number of changed edges in $\VD(S)$ is constant, offering some hope that an efficient dynamic---or at least semi-dynamic---data structure for maintaining $\VD(S)$ could exist.
These hopes, however, are rapidly squashed, as it is easy to construct examples where the complexity of each successively added face is $\Omega(n)$, and thus each insertion changes the position of a linear number of vertices and edges of $\VD(S)$.
The goal of this paper is to show that despite this worst-case behavior, the amortized number of \emph{structural changes} to the graph of the Voronoi diagram of $S$, that is, the minimum number of edge insertions and deletions needed to update $\VD(S)$ throughout any sequence of site insertions to $S$, is much smaller.

This might come as a surprise in light of the fact that the number of combinatorial changes (usually modeled as \emph{flips}) to the Delaunay triangulation of $S$ upon the insertion of a point can be $\Omega(n)$ with each insertion, 
even when the sites are in convex position and are added in clockwise order. (Note that in that case the Voronoi diagram of $S$ is a tree and the standard flip operation is a rotation in the tree.)

To overcome this worst-case behavior, \citet{aronov2006data} studied what happens in this specific case (points in convex position added in clockwise order) if the rotation operation is replaced by the more elementary link (add an edge) and cut (delete an edge) operations in the tree. 
They show that, in that model, it is possible to reconfigure the tree after each site insertion while performing $O(\log n)$  links and cuts, amortized; however their proof is existential and no algorithm is provided to find those links and cuts. Pettie~\cite{pettie} shows both an alternate proof of that fact using forbidden 0-1 matrices and a matching lower bound.

One important application of Voronoi diagrams is to solve \emph{nearest-neighbor} (or \emph{farthest-neighbor}) queries: given a point in the plane, find the site nearest (or farthest) to this point. In the static case, this is done by preprocessing the (nearest or farthest point) Voronoi diagram to answer point-location queries in $O(\log n)$ time.
Without the need to maintain $\VD(S)$ explicitly, the problem of nearest neighbor queries is a \emph{decomposable search problem} and can be made semi-dynamic using the standard dynamization techniques of~\citet{bentley1980decomposable}.   
The best incremental data structure supporting nearest-neighbor queries performs queries and insertions in $O(\log^2 n/ \log \log n)$ time~\cite{chiang1992dynamic,overmars1983design}. 
Recently, \citet{chan2010dynamic} developed a randomized data structure supporting nearest-neighbor queries in $O(\log^2 n)$ time, insertions in $O(\log^3 n)$ expected amortized time, and deletions in $O(\log^6 n)$ expected amortized time.

\subsection*{Flarbs}
In the mid-1980's it was observed that a number of variants of Voronoi diagrams and Delaunay triangulations using different metrics (Euclidean distance, $L_p$ norms, convex distance functions) or different kinds of sites (points, segments, circles) could all be handled using similar techniques. 
To formalize this, several abstract frameworks were defined, such as the one of Edelsbrunner and Seidel~\cite{voroarr} and the two variants of abstract Voronoi diagrams of Klein~\cite{Klein2009885, DBLP:books/sp/Klein89}. 
In this paper we define a new abstract framework to deal with Voronoi diagrams constructed incrementally by inserting new sites.

Let $G$ be a 3-regular embedded plane graph with $n$ vertices\footnote{While the introduction used $n$ for the number of sites in $S$, the combinatorial part of this article uses $n$ for the number of vertices in the Voronoi diagram. By Euler's formula, those two values are asymptotically equivalent, up to a constant factor.}. 
We seek to bound the number of edge removals and insertions needed to implement the following operation, hereafter referred to as a \emph{flarb}\footnote{Although the last two authors are honored by the flattering renaming of the flarb operation in the literature \cite{pettie}, this paper uses original terminology.}: 
Given a simple closed curve $\curve$ in the plane whose interior intersects $G$ in a connected component, split both $\curve$ and all the edges that it crosses at the point of intersection, remove every edge and vertex that lies in the interior of $\curve$, and add each curve in the subdivision of $\curve$ as a new edge; see Figure~\ref{fig:FlarbExample}.
This operation can be used to represent the insertion of new cells in different types of Voronoi diagrams. 
It can also be used to represent the changes to the 1-skeleton of a polyhedron 
in $\mathbb{R}^3$ after it is intersected with a halfspace.

\begin{figure}[h]
\centering
\includegraphics{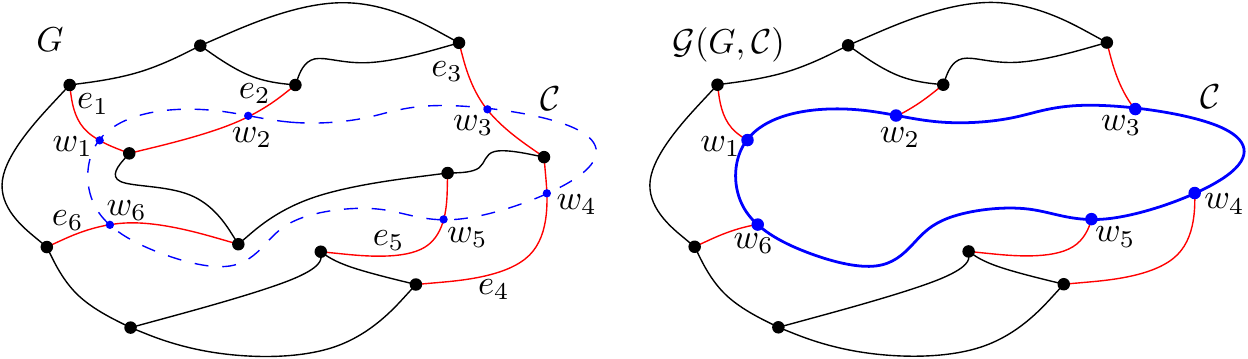}
\caption{\small The flarb operation on a graph $G$ induced by a flarbable curve $\curve$, produces a graph $\calG(G, \curve)$ with 2 more vertices. Fleeq-edges crossed by $\curve$ are shown in red.}
\label{fig:FlarbExample}
\end{figure}

\subsection*{Results}
We show that the amortized \emph{cost} of a flarb operation, where the combinatorial cost is defined to be the minimum number of edge insertions and removals needed to perform it, is $O(\sqrt{n})$. We also show a matching lower bound: some sequences of flarbs require $\Omega(\sqrt{n})$ links and cuts per flarb, even when the graph is a tree (or more precisely a Halin graph---a tree with all leaves connected by a cycle to make it 3-regular). This contrasts with the $O(\log{n})$ upper bound of \citet{aronov2006data} for the Voronoi diagram of points in convex position (also a tree) when points are added in clockwise order.

We complement these combinatorial bounds with an algorithmic result. 
We present an output-sensitive data structure that maintains the nearest-~or farthest-point Voronoi diagram of a set $S$ of $n$ points in convex position as new points are added to $S$. Upon an insertion, the data structure finds the minimum (up to within a constant factor) number $K$ of edge insertions and deletions necessary to update the Voronoi diagram of $S$.

The running time of each insertion is $O(K \log^7 n)$, and by our combinatorial bounds, $K=O(\sqrt{n})$. This solves the open problem posed by~\citet{aronov2006data}.

The distinguishing feature of this data structure is that it explicitly maintains the graph structure of the Voronoi diagram after every insertion, a property that is not provided by any nearest neighbor data structure that uses decomposable searching problem techniques. 
Further, the data structure also maintains the Voronoi diagram in a \emph{grappa tree} \cite{aronov2006data}, a variant of the link-cut tree of Sleator and Tarjan \cite{sleator1983data} which allows a powerful query operation called \emph{oracle-search}. Roughly speaking, the oracle-search query has access to an oracle specifying a vertex to find. Given an edge of the tree, the oracle determines which of the two subtrees attached to its endpoints contains that vertex. Grappa trees use $O(\log n)$ time and oracle calls to find the sought vertex. A grappa tree is in some sense a dynamic version of the centroid decomposition for trees, which is used in many algorithms for searching in Voronoi diagrams. 
Using this structure, it is possible to solve a number of problems for the set $S$ at any moment during the incremental construction, for example:
\begin{itemize}
\item Given sites $p$ and $q$,  report whether they are connected by a Delaunay edge in $O(\log n)$ time.
\item Given a point $q$, find the Voronoi cell containing $q$ in $O(\log{n})$ time. This not only gives the nearest neighbor of $q$, but a pointer to the explicit description of its cell.
\item Find the smallest disk enclosing $S$, centered on a query segment $[pq]$, in $O(\log n)$ time \cite{enclosingdisk_CCCG08}.
\item Find the smallest disk enclosing $S$, centered on a query circle $C$, in $O(\log n)$ time \cite{DiskConstrainedOneCenterQueries}.
\item Given a convex polygon $P$ (counterclockwise array of its $m$ vertices), find the smallest disk enclosing $S$ and excluding $P$ in $O(\log n + \log m)$ time \cite{NewCircleSeparability}.
\end{itemize}

The combinatorial bound for Voronoi diagrams also has direct algorithmic consequences, the most important being that it is possible to store all versions of the graph throughout a sequence of insertions using persistence in $O(n^{3/2})$ space.
Since the entire structure of the graph is stored for each version, this provides a foundation for many applications that, for instance would require searching the sequence of insertions for the moment during which a specific event occurred. 
 
\subsection*{Outline}
The main approach used to bound the combinatorial cost of a flarb is to examine how the complexity of the faces changes. Notice that faces whose size remains the same do not require edge insertions and deletions. The other faces either grow or shrink, and a careful counting argument reveals that the cost of a flarb is at most the number faces that shrink (or disappear)  upon execution of the flarb (Section~\ref{sec: The flarb operation}). By using a potential function that sums the sizes of all faces, the combinatorial cost of shrinking faces is paid for by the reduction of their potential.
To avoid incurring a high increase in potential for a large new face, the potential of each face is capped at $\sqrt{n}$. Then at most $O(\sqrt{n})$ large faces can shrink without changing potential and are accounted for separately (Section~\ref{sec:Combinatorial Bound}).
The matching $\Omega(\sqrt{n})$ lower bound is presented in Section~\ref{sec:Lower Bound}, and Section~\ref{sec:Computing the flarb} presents the data structure for performing flarbs for the Voronoi diagrams of points in convex position.

\section{The flarb operation}\label{sec: The flarb operation}

In this section we formalize the flarb operation that models the insertion of new sites in Voronoi diagrams and present a preliminary analysis of the cost of a flarb. 

Let $G=(V,E)$ be a planar 3-regular graph embedded in $\mathbb{R}^2$ (not-necessarily with a straight-line embedding).
Let $\curve$ be a simple closed Jordan curve in the plane. Define $\ins$ to be the set of vertices of $G$ that lie in the interior of $\curve$ and let $\out = V\setminus \ins$.
We say that $\curve$ is \emph{flarbable} for $G$ if the following conditions hold:
\shortfull{}{\begin{enumerate}}
\shortfull{(1)}{\item} the graph induced by $\ins$ is connected, 
\shortfull{(2)}{\item} $\curve$ intersects each edge of $G$ either at a single point or not at all,
\shortfull{(3)}{\item} $\curve$ passes through no vertex of $G$, and
\shortfull{(4)}{\item} the intersection of $\curve$ with each face of $G$ is path-connected. 
\shortfull{}{\end{enumerate}}

In the case where the graph $G$ is clear from context, we simply say that 
$\curve$ is flarbable.
The \emph{fleeq} of $\curve$ is the circular sequence $\fleeq = e_1, \ldots, e_k$ of edges in $E$ that are crossed by $\curve$; we call the edges in $\fleeq$ \emph{fleeq-edges}. A face whose interior is crossed by $\curve$ is called a \emph{$\curve$-face}.
We assume without loss of generality that $\curve$ is oriented clockwise and that the edges in $\fleeq$ are ordered according to their intersection with $\curve$.
Given a flarbable curve~$\curve$ on $G$, we present the following definition.

\begin{definition} \label{def: flarb}
 For a planar graph $G$ and a curve $\curve$ that is flarbable for $G$, we define a \emph{flarb} operation $\flarb$ which produces a new 3-connected graph $\calG(G, \curve)$ as follows (see Figure~\ref{fig:FlarbExample} for a depiction):
\shortfull{}{\begin{enumerate}}
\shortfull{(1)}{\item} For each edge $e_i = (u_i, v_i)$ in $\fleeq$ such that $u_i \in \ins$ and $v_i \in \out$, create a new vertex $w_i = \curve\cap e_i$ and connect it to $v_i$ along $e_i$.  
\shortfull{(2)}{\item}
For each pair $e_i, e_{i+1}$ of successive edges in $\fleeq$, create a new edge $(w_i, w_{i+1})$ between them along $\curve$. We call $(w_i, w_{i+1})$ a \emph{$\curve$-edge} (all indices are taken modulo $k$).
\shortfull{(3)}{\item}
Delete all vertices of $\ins$ along with their incident edges.
\shortfull{}{\end{enumerate}}
\end{definition}

\begin{lemma}\label{lemma:Net change in vertices after flarb}
For each flarbable curve $\curve$ on a 3-regular planar graph $G$, $\calG(G, \curve)$ has at most 2 more vertices than $G$ does.
\end{lemma}
\begin{proof}
Let $\fleeq = e_1, \ldots, e_k$ be the fleeq of $\curve$ and let $f$ be the new face in $\calG(G, \curve)$ that is bounded by $\curve$ and created by the flarb operation $\flarb$.
Notice that the vertices of $f$ are the points $w_1, \ldots, w_k$ along edges $e_1, \ldots, e_k$, where $w_i = \curve\cap e_i$.
Since $\curve$ is flarbable, the subgraph induced by the vertices of $\ins\cup \{w_1, \ldots, w_k\}$ is also a connected graph $T$ with $w_1, \ldots, w_k$ as its leaves and every other vertex of degree 3; see Figure~\ref{fig:FlarbExample}.  
Therefore $T$ has at least $k-2$ internal vertices.
The flarb operation adds $k$ vertices, namely $w_1, \ldots, w_k$, and the internal vertices of $T$ are deleted. Therefore, the net increase in the number of vertices is at most 2.
\end{proof}

\shortfull{Note that at most two new vertices are created (see appendix).}{}
Since each newly created vertex has degree three and all remaining vertices are unaffected, the new graph is 3-regular. 
In other words, the flarb operation $\flarb$ creates a cycle along $\curve$ and removes the portion of the graph enclosed by $\curve$. 
Note that for any point set in general position (no four points lie on the same circle), its Voronoi diagram is a 3-regular planar graph, assuming we use the line at infinity to join the endpoints of its unbounded edges in clockwise order. Therefore, a flarb can be used to represent the changes to the Voronoi diagram upon insertion of a new site.

\begin{observation}\label{lemma:Voronoi insertion is flarbable}
Given a set $S$ of points in general position, let $\calV(S)$ be the graph of the Voronoi diagram of $S$. 
For a new point $q$, there exists some curve $\curve^q_S$ such that $\calG(\calV(S), \curve^q_S) = \calV(S\cup \{q\})$; namely, $\curve^q_S$ is the boundary of the Voronoi cell of $q$ in $\calV(S\cup \{q\})$.
\end{observation}

More generally, convex polytopes defined by the intersection of halfspaces in $\mathbb{R}^3$ behave similarly: the intersection of a new halfspace with a convex polytope modifies the structure of its 1-skeleton by adding a new face. This structural change can be implemented by performing a flarb operation in which the flarbable curve consists of the boundary of the new face.

\subsubsection*{Preserved faces and edges}
\begin{definition} \label{def:preserved}
Given a $\curve$-face $f$ of $G$,
the \emph{modified face} of $f$ is the face $f'$ of $\calG(G, \curve)$ that coincides with $f$ outside of $\curve$.
In other words, $f'$ is the face that remains from $f$ after performing the flarb $\flarb$.
We say that a $\curve$-face $f$ is \emph{preserved} (by the flarb $\flarb$) if $|f| = |f'|$. 
Moreover, we say that each edge in a preserved face is \emph{preserved} (by $\flarb$).
Denote by $\pres$ the set of faces preserved by $\flarb$ and let $\bounded$ be the set of faces wholly contained in the interior of $\curve$.
\end{definition}
Since a preserved $\curve$-face bounded by two fleeq-edges $e_i$ and $e_{i+1}$ has the same size before and after the flarb, 
there must be an edge $e$ of $G$ connecting $e_i$ with $e_{i+1}$ which is replaced by a $\curve$-edge $e^*$ after the flarb. 
In this case, we say that the edge $e$ \emph{reappears} as $e^*$.

The following auxiliary lemma will help us bound the number of operations needed to produce the graph $\calG(G, \curve)$, and follows directly from the Euler characteristic of connected planar graphs:

\begin{lemma}\label{lemma:size of graphs with holes}
Let $H$ be a connected planar graph with vertices of degree either 1, 2 or 3. For each $i\in \{1, 2, 3\}$, let $\delta_i$ be the number of vertices of $H$ with degree $i$. 
Then, $H$ has exactly $2 \delta_1 + \delta_2 + 3F_H-3$ edges, where $F_H$ is the number of bounded faces of $H$.
\end{lemma}

\subsection{Combinatorial cost of a flarb}
Given a 3-regular graph $G = (V, E)$ and a flarbable curve $\curve$ we want to analyze the number of structural changes that $G$ must undergo to perform $\flarb$.
To this end, we define the combinatorial \emph{cost} of $\flarb$, denoted by $\cost$, to be the minimum number of links and cuts needed to transform $G$ into $\calG(G, \curve)$ (note that the algorithm may not implement the flarb operation according to the procedure described in Definition~\ref{def: flarb}).
We assume that any other operation has no cost and is therefore not included in the cost of the flarb. 

Consider the fleeq $\fleeq = e_1, \ldots, e_k$ and the $\curve$-edges created by $\flarb$.  
Let $e$ be an edge adjacent to some $e_i$ and $e_{i+1}$ that reappears as the $\curve$-edge $e^*$.
Notice that we can obtain $e^*$ without any links or cuts to $G$: simply shrink $e_i$ and $e_{i+1}$ so that their endpoints in $\ins$ now coincide with their intersections with $\curve$. Then modify $e$ to coincide with the portion of $\curve$ connecting the new endpoints of $e_i$ and $e_{i+1}$.
Using this \emph{preserving operation}, we obtain the $\curve$-edge $e^*$ with no cost to the flarb.
Intuitively, preserved edges are cost-free in a flarb while non-preserved edges have a nonzero cost. This notion is formalized in the following lemma.

\begin{lemma}\label{lemma:The cost of a flarb}
For a flarbable curve $\curve$, $$(|\fleeq| + |\bounded| - |\pres|)/2 \leq \cost  \leq 4|\fleeq| + 3|\bounded| - 4|\pres|.$$
\end{lemma}
\begin{proof}
For the upper bound, we describe a construction of $\calG(G, \curve)$ from $G$ using at most $|\fleeq| + 3|\bounded| - 4|\pres$ links and cuts\footnote{We caution the reader that while this construction is algorithmic in nature, it is used purely to provide an upper-bound and does not reflect the behavior of the algorithm presented in Section~\ref{sec:Computing the flarb} that gives our desired runtime.}.
Consider the subgraph $\calG_\curve$ induced by $\ins \cup \{v : v\textrm{ is an endpoint of some edge in } \fleeq \}$. 
Since~$\curve$ is flarbable,  $\calG_\curve$ is a connected  graph such that each vertex of $\ins$ has degree 3 while the endpoints of the fleeq-edges outside of $\curve$ have degree 1. 
Note that if two preserved faces share a non-fleeq edge $e$, then there are four neighbors of the endpoints of $e$ that lie outside of $\curve$. 
Since $\calG_\curve$ is connected, $e$ and its four adjacent edges define the entire graph $\calG_\curve$ and the bound holds trivially. Therefore, we assume that no two preserved faces share a non-fleeq-edge from this point forward.

Note that the bounded faces of $\calG_\curve$ are exactly the bounded faces in $\bounded$.
Since $\calG_\curve$ has $|\fleeq|$ vertices of degree 1, no vertices of degree 2, and $|\bounded|$ bounded faces, by Lemma~\ref{lemma:size of graphs with holes}, $\calG_\curve$ has at most $2|\fleeq| + 3|\bounded|$ edges.
Every edge of $\calG_\curve$ that is not preserved is removed with a cut operation (isolated vertices will be removed afterwards). 
Note that each preserved face contains at least three preserved edges: two fleeq-edges and a third edge of $G$. Based on the assumption that no two preserved faces share a non-fleeq-edge, the third edge is not double counted, while the fleeq-edges may be counted at most twice. Therefore, each preserved face contributes at least two preserved edges that are specific to that face, meaning that a total of at most $2|\fleeq| + 3|\bounded|- 2|\pres|$ cut operations are performed. 
Note that each non-preserved fleeq-edge has been cut and will need to be reintroduced later to obtain $\calG(G, \curve)$.

Recall that no edge bounding a preserved face has been cut.
For each preserved face, perform a preserving operation on it which requires no link or cut operation. Since no two preserved faces share a non-fleeq edge, all the $\curve$-edges bounding the preserved faces are added without increasing \cost.
To complete the construction of $\calG(G, \curve)$, create each fleeq-edge that is not preserved and then add the remaining $\curve$-edges bounding non-preserved $\curve$-faces.
Because at least $|\pres|$ fleeq-edges were preserved, at most $|\fleeq| - |\pres|$ fleeq-edges must be reintroduced.
Moreover, since only $|\fleeq| - |\pres|$ $\curve$-faces are not preserved, we need to create at most $|\fleeq| - |\pres|$ $\curve$-edges. 
Therefore, this last step completes the flarb and construct $\calG(G, \curve)$ using a total of at most $2|\fleeq| - 2|\pres|$ link operations. 
Consequently, the total number of link and cuts needed to obtain $\calG(G, \curve)$ from $G$ is at most $4|\fleeq| + 3|\bounded| - 4|\pres|$ as claimed. 

To show that $\cost > (|\fleeq| + |\bounded| - |\pres|)/2$, simply note that in every non-preserved $\curve$-face, the algorithm needs to perform at least one cut, either to augment the size or reduce the size of the face.
Because $G$ and $\curve$ define exactly $|\fleeq| + |\bounded|$ faces, 
and since in all but $|\pres|$ of them at least one of its edges must be cut, 
at least $|\fleeq| + |\bounded| - |\pres|$ edges must be cut. 
Since an edge belongs to at most two faces a cut can be over-counted at most twice and the claimed bound holds.
\end{proof}

\section{The combinatorial upper bound}\label{sec:Combinatorial Bound}
In this section, we define a potential function to bound the amortized cost of each operation in a sequence of flarb operations. 
For a 3-regular embedded planar graph $G = (V,E)$, we define two potential functions: a local potential function $\mu$ to measure the potential of each face, and a global potential function $\Phi$ to measure the potential of the whole graph.

\begin{definition}
Let $F$ be the set of faces of a 3-regular embedded planar graph $G = (V, E)$.  For each face $f \in F$, let $\mu(f) = \min\{\lceil\sqrt{|V|}\rceil, |f|\}$, where $|f|$ is the number of edges on the boundary of $f$.  
The potential $\Phi(G)$ of $G$ is defined as follows:
\shortfull{$}{\[}
\Phi(G) = \lambda\sum_{f \in F}\mu(f),
\shortfull{$}{\]}
for some sufficiently large positive constant $\lambda$ to be defined later.
\end{definition}

Recall that the potential $\mu(f)$ of a $\curve$-face $f$ remains unchanged as long as $|f|, |f'|\geq \sqrt{|V|}$, where $f'$ is the modified face of $f$ after the flarb. 
Since there is no change in potential that we can use within large $\curve$-faces, 
we exclude them from our analysis and focus only on smaller $\curve$-faces. We formalize this notion in the following section.

\subsection{Flarbable sub-curves}

Given a flarbable curve $\curve$, a (connected) curve $\gamma \subseteq \curve$ is a \emph{flarbable sub-curves}.
Let $\eps_\gamma = e_1, \ldots, e_k$ (or simply~$\eps$) be the set of fleeq-edges intersected by $\gamma$ given in order of intersection after orienting $\gamma$ arbitrarily.
We call $\eps$ the \emph{subfleeq} induced by $\gamma$.
We say that a face is a $\gamma$-face if two of its fleeq-edges are crossed by $\gamma$ (if $\gamma$ has an endpoint in the interior of this face, it is not a $\gamma$-face).

Consider the set of all edges of $G$ intersected or enclosed by $\curve$ that bound some $\gamma$-face.
Since $\fleeq$ is flarbable, these edges induce a connected subgraph $Y_\gamma$ of $G$ with $|\eps| = k$ leaves (vertices of degree 1), namely the endpoints outside of $\curve$ of each fleeq-edge in $\eps$; see Figure~\ref{fig:Subflarbable curves}.
Notice that $Y_\gamma$ may consist of some bounded faces contained in the interior of $\curve$.
Let $H_\gamma$ be the set of bounded faces of $Y_\gamma$ and let $\delta_2$ be the number of vertices of degree 2 of $Y_\gamma$.
Since $Y_\gamma$ consists of $k$ vertices of degree 1, Lemma~\ref{lemma:size of graphs with holes} implies the following result.

\begin{figure}[tb]
\centering
\includegraphics{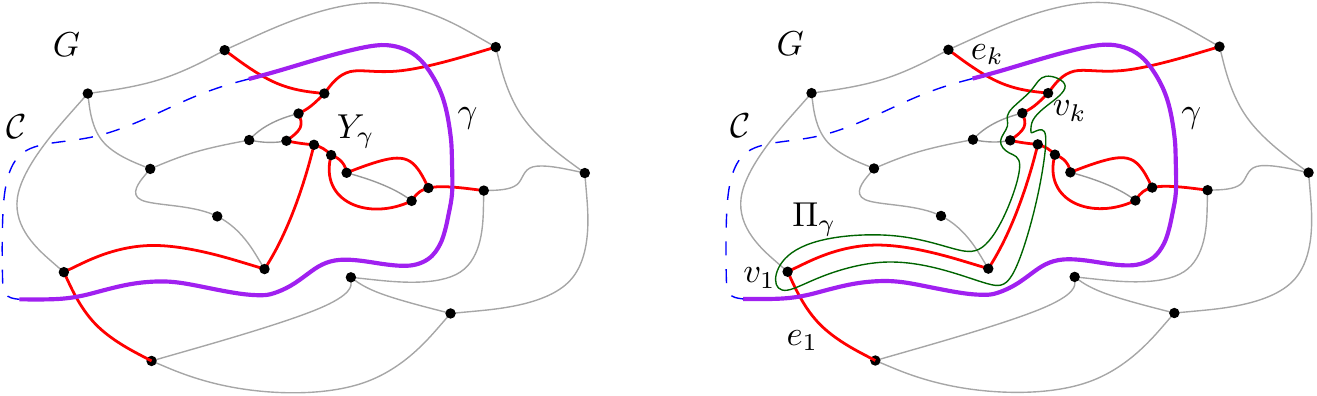}

\caption{\small Left: A flarbable sub-curves $\gamma$ is contained in a flarbable curve $\curve$. 
The graph $Y_\gamma$ is the union of all edges bounding a  $\gamma$-face.
Right: The path $\Pi_\gamma$ connects the endpoints of the first and last fleeq-edges crossed by $\gamma$ by going along the boundary of the outer-face of $Y_\gamma$.}
\label{fig:Subflarbable curves}
\end{figure}

\begin{corollary}\label{corollary:Size of Y}
The graph $Y_\gamma$ consists of exactly $2k + \delta_2 + 3|H_\gamma| - 3$ edges.
\end{corollary}

Recall that a $\curve$-face $f$ is preserved if its corresponding modified face $f'$ in $\calG(G, \curve)$ has the same number of edges, i.e., $|f'| = |f|$.
We say that $f$ is \emph{augmented} if $|f'| = |f| + 1$ and
we call $f$ \emph{shrinking} if $|f'| < |f|$.  
Notice that these are all the possible cases as $f$ gains at most one new edge during the flarb, namely the $\curve$-edge crossing this face.

In the context of a particular flarbable sub-curve $\gamma$,
let $a_\gamma,s_\gamma$ and $p_\gamma$ be the number of augmented, shrinking and preserved $\gamma$-faces, respectively (or simply $a,s$ and $p$ if $\gamma$ is clear from the context). 
We further differentiate among the $s$ shrinking $\gamma$-faces.
A shrinking $\gamma$-face is \emph{interior} if it contains no vertex of degree~2 of $Y_\gamma$ and does not share an edge with an augmenting face.
Let $s_a$ be the number of shrinking $\gamma$-faces that share an edge with an augmented face, let $s_b$ be the number of shrinking $\gamma$-faces not adjacent to an augmented face that have a vertex of degree 2 of $Y_\gamma$, and let $s_c$ be the number of interior shrinking $\gamma$-faces. 
Therefore, $s = s_a + s_b + s_c$ is the total number of shrinking $\gamma$-faces.

Since each augmented face has at most two edges and because there are $a$ augmented faces, we know that $s_a \leq 2a$.
Let $v_1$ and $v_k$ be the endpoints of the edges $e_1$ and $e_k$ that lie inside $\curve$.
Let $\Pi_{\gamma}$ be the unique path connecting $v_1$ and $v_k$ in $Y_\gamma$ that traverses the boundary of the outer face of $Y_\gamma$ and stays in the interior of $\curve$; see Figure~\ref{fig:Subflarbable curves}.

Notice that $\Pi_\gamma$ contains all the edges of $\gamma$-faces that may bound a $\gamma'$-face for some other flarbable sub-curve $\gamma'$ disjoint from $\gamma$.
In the end, we aim to have bounds on the number of edges that will be removed from the $\gamma$-faces during the flarb, but some of these edges may be double counted if they are shared with a $\gamma'$-face. 
Therefore, we aim to bound the length of $\Pi_\gamma$ and count precisely the number of edges that could possibly be double-counted. 

\begin{lemma}\label{lemma:Bounding the length of Pi}
The path $\Pi_{\gamma}$ has length at most $k + 3|H_\gamma| + \delta_2 - a - s_c$.
\end{lemma}

\begin{proof}
Notice that no fleeq-edge can be part of $\Pi_{\gamma}$ or this path would go outside of $\curve$, i.e., there are $k$ fleeq-edges of $Y_\gamma$ that cannot be part of $\Pi_{\gamma}$.

We say that a vertex is \emph{augmented} if it is incident to two fleeq-edges and a third edge that is not part of~$\eps$, which we call an \emph{augmented edge}. Because each augmented $\gamma$-face has exactly one augmented vertex, there are exactly $a$ augmented vertices in $Y_\gamma$. 
Moreover, $\Pi_{\gamma}$ contains at most 2 augmented vertices (if $v_1$ or $v_k$ is augmented). 
Thus, at most two augmented edges can be traversed by $\Pi_{\gamma}$ and hence, at least $a-2$ augmented edges of $Y_\gamma$ do not belong to~$\Pi_{\gamma}$. 

Let $f$ be an internal shrinking $\gamma$-face. 
Since $f$ is not adjacent to an augmented $\gamma$-face, it has no augmented edge on its boundary. We claim that $f$ has at least one edge that is not traversed by $\Pi_{\gamma}$. 
If this claim is true, then there are at least $s_c$ non-fleeq non-augmented edges that cannot be used by $\Pi_{\gamma}$---one for each internal shrinking $\gamma$-face.
Thus, since $Y_\gamma$ consists of $2k + 3|H_\gamma| + \delta_2 - 2$ edges, the number of edges in $\Pi_{\gamma}$ is at most
\shortfull{
$2k + 3|H_\gamma| + \delta_2 - 2 - (k + a-2 + s_c) = k + 3|H_\gamma| + \delta_2 - a - s_c.$
}{\[
2k + 3|H_\gamma| + \delta_2 - 2 - (k + a-2 + s_c) = k + 3|H_\gamma| + \delta_2 - a - s_c.
\]
}

It remains to show that each internal shrinking $\gamma$-face $f$ has at least one non-fleeq edge that is not traversed by $\Pi_{\gamma}$. 
If $\Pi_{\gamma}$ contains no edge on the boundary of $f$, then the claim holds trivially.
If $\Pi_{\gamma}$ contains exactly one edge of $f$, then since $f$ is shrinking, it has at least 4 edges and two of them are not fleeq-edges. Thus, in this case there is one edge of $f$ that is not traversed by $\Pi_{\gamma}$. We assume from now on that $\Pi_\gamma$ contains at least two edges of $f$.

We claim that that $\Pi_{\gamma}$ visits a contiguous sequence of edges along the boundary of $f$. 
To see this, note that each face of $Y_\gamma$ lying between $\Pi_{\gamma}$ and the boundary of $f$ cannot be crossed by $\curve$. Therefore, if we consider the first edge of $\Pi_{\gamma}$ that is not on the boundary of $f$ after visiting $f$ the first time, the this edge is incident to the outer face of $Y_\gamma$ and the only face of $Y_\gamma$ that it is incident with does not intersect~$\curve$. This is a contradiction, since this edge should not be part of $Y_\gamma$ by definition.
Therefore, $\Pi_{\gamma}$ visits a contiguous sequence of edges along $f$. 

If $\Pi_{\gamma}$ visits 2 consecutive edges of $f$, then the vertex in between them must have degree 2 in $Y_\gamma$, as the two edges are incident to the outer face---a contradiction since $f$ is an internal shrinking face with no vertex of degree 2. Consequently, if $f$ is an internal shrinking face, it has always at least one non-fleeq edge that is not traversed by $\Pi_{\gamma}$.
\end{proof}

\subsection{How much do faces shrink in a flarb?}

In order to analyze the effect of the flarb operations on flarbable sub-curves, we think of each edge as consisting of two \emph{half-edges}, each adjacent to one of the two faces incident to this edge. 
For a given edge, the algorithm may delete its half-edges during two separate flarbs of differing flarbable sub-curves.

We define the operation $\flarb[G, \gamma]$ to be the operation which executes steps 1 and 2 of the flarb on the flarbable sub-curve $\gamma$ and then deletes each half-edge with both endpoints in $\ins$ adjacent to a $\gamma$-face. 
Since $\flarb[G, \gamma]$ removes and adds half-edges, we are interested in bounding the net balance of half-edges throughout the flarb. 
To do this, we measure the change in size of a face during the flarb.

Recall that $a,s$ and $p$ are the number of augmented, shrinking and preserved $\gamma$-faces, respectively. 
The following result provides a bound on the total ``shrinkage'' of the faces crossed by a given flarbable sub-curve. 

\begin{theorem} \label{thm:delta}
Given a flarbable curve $\curve$ on $G$ and a flarbable sub-curve $\gamma$ crossing the fleeq-edges $\eps = e_1, \ldots, e_k$,
let $f_1, \ldots, f_{k}$ be the sequence of $\gamma$-faces and let $f_1', \ldots, f_{k}'$ be their corresponding modified faces after the flarb $\flarb[G, \gamma]$. Then,
\shortfull{
$\sum_{i = 1}^k (|f_i| - |f_i'|) \geq s/2$.
}{
\begin{equation} \label{eqn:deltathm}
\sum_{i = 1}^k (|f_i| - |f_i'|) \geq s/2.
\end{equation}
}
\end{theorem}
\begin{proof}
Recall that no successive pair of $\gamma$-faces can both be augmented unless $\fleeq$ consists of three edges incident to a single vertex.
In this case, at most $3$ $\gamma$-faces can be augmented, so $\sum_{i = 1}^k (|f_i| - |f_i'|) = 3$ and the result holds trivially; hence, we assume from now on that no two successive faces are both augmented.

Let $\Delta$ be the number of half-edges removed during $\flarb[G, \gamma]$.
Notice that to count how much a face $f_i$ shrinks when becoming $f'_i$ after the flarb, we need to count the number of half-edges of $f_i$ that are deleted and the number that are added in $f'_i$. 
Since exactly one half-edge is added in each $f'_i$, we know that $\sum_{i = 1}^k (|f_i| - |f_i'|) = \Delta - k$.
We claim that $\Delta \geq k + s/2$.
If this claim is true, then $\sum_{i = 1}^k (|f_i| - |f_i'|) \geq s/2$ as stated in the theorem.
In the remainder of this proof, we show this bound on $\Delta$.

Let $\calT = (V_\calT, E_\calT)$ be the subgraph of $Y_\gamma$ obtained by removing its $k$ fleeq-edges. 
It follows from~\ref{corollary:Size of Y} that $|E_\calT|  = k + 3|H_\gamma| + \delta_2 - 3$ .
To obtain a precise counting of $\Delta$, notice that for some edges of $\calT$, $\flarb[G, \gamma]$ removes only one of their half-edges and for others it will remove both of them. 
Since the fleeq-edges are present in each of the faces $f_1, \ldots, f_k$ before and after the flarb, 
we get that
\shortfull{$\Delta = 2|E_\calT| - S_{\calT},$}{
\begin{equation} \label{eqn:single-edges}
\Delta = 2|E_\calT| - S_{\calT},
\end{equation}
} where $S_{\calT}$ denotes the number of edges in $\calT$ with only one half-edge incident to a face of $f_1, \ldots, f_k$. 

Note that the edges of $S_{\calT}$ are exactly the edges on the path $\Pi_\gamma$ bounded in Lemma~\ref{lemma:Bounding the length of Pi}.
Therefore, $S_{\calT} \leq k + 3|H_\gamma| + \delta_2 - a - s_c$. 
By using this bound\shortfull{}{ in~(\ref{eqn:single-edges})}, we get 
\[
\Delta \geq 2(k+3|H_\gamma| + \delta_2 -3) - (k + 3|H_\gamma| + \delta_2 - a - s_c)  = 
k + 3|H_\gamma| + \delta_2 +a + s_c - 6.
\]

Since each shrinking $\gamma$-face accounted for by $s_b$ has a vertex of degree 2 in $Y_\gamma$, we know that $\delta_2 \geq s_b$. Moreover, $s_a\leq 2a$ as each shrinking $\gamma$-face can be adjacent to at most two augmenting $\gamma$-faces. Therefore, since $s = s_a+ s_b + s_c$, we get that $\Delta \geq k + 3|H_\gamma| +s_a/2 + s_b + s_c\geq k + s/2$, where $s$ is the number of shrinking $\gamma$-faces proving the claimed bound on $\Delta$.
\end{proof}

\subsection{Flarbable sequences}
Let $\calG^0  =  G$.
A sequence of curves $\scrC = \curve_1, \ldots, \curve_k$ is \emph{flarbable} if for each $i \in [k]$, $\calC_i$ is a flarbable on 
\shortfull{$}{\[}
\calG^i = \calG(\calG^{i-1}, \curve_{i}).
\shortfull{$}{\]}
As a notational shorthand, let $\calF^i$ denote the flarb operation $\calF(\calG^{i-1}, \curve_i)$ when $\scrC$ is a flarbable sequence for~$G$.  

\begin{theorem} \label{thm:amortizeoneflarb}
For a 3-regular planar graph $G = (V, E)$ and some flarbable sequence $\scrC = \curve_1, \ldots, \curve_N$ of flarbable fleeqs, for all $i \in [N]$,
\shortfull{$}{$$}\cost[\calG^{i-1}, \curve_i] + \Phi(\calG^i) - \Phi(\calG^{i-1}) \leq O(\sqrt{|V_i|}),\shortfull{$}{$$}
where $V_i$ is the set of vertices of $\calG^i$.
\end{theorem}
\begin{proof}
Partition $\curve_i$ into smaller curves $\gamma_1, \ldots, \gamma_h$ such that for all $j \in [h]$, $\gamma_j$ is a maximal curve contained in $\curve_i$ that does not intersect the interior of a face with more than $\sqrt{|V_i|}$ edges. 
Since there can be at most $\sqrt{|V_i|}$ faces of size $\sqrt{|V_i|}$, we know that $h \leq \sqrt{|V_i|}$.
Let $\eps_j$ be the subfleeq containing each fleeq-edge crossed by $\gamma_j$.
Let $a_j,s_j$ and $p_j$ be the number of augmented, shrinking and preserved $\gamma_j$-faces, respectively. Notice that $|\eps_j| = a_j + s_j + p_j+1$. Moreover, since each augmented face is adjacent to a shrinking face, we know that $a_j \leq s_j+1$. Therefore, $|\eps_j| \leq  2s_j + p_j+2$.

Let $\calL_i$ be the set of $\curve_i$-faces with at least $\sqrt{|V_i|}$ edges and let $\omega_i$ be the set of all faces of $\calG^{i-1}$ completely enclosed in the interior of~$\curve_i$. 

First, we upper bound $\cost[\calG^{i-1}, \curve_i]$.  
By Lemma~\ref{lemma:The cost of a flarb}, we know that

\begin{align}
\cost[\calG^{i-1}, \curve_i] 
&\leq 4|\fleeq[\curve_i]| + 3|\bounded[\calG^{i-1}, \curve_i]| - 4|\pres[\calG^{i-1}, \curve_i]|\\
&= 4\sum_{j  = 1}^h |\eps_j| + 3|\bounded[\calG^{i-1}, \curve_i]| - 4|\pres[\calG^{i-1}, \curve_i]|\\
&\leq 4\sum_{j = 1}^h (2s_j + p_j+2) + 3|\bounded[\calG^{i-1}, \curve_i]| - 4|\pres[\calG^{i-1}, \curve_i]|
\end{align}

\noindent because each preserved face is crossed by exactly one flarbable sub-curve, $\sum_{j = 1}^h p_j = |\pres[\calG^{i-1}, \curve_i]|$. Therefore, 
\shortfull{$}{\[}
\cost[\calG^{i-1}, \curve_i] \leq 4\sum_{j = 1}^h (2s_j + 2) + 3|\bounded[\calG^{i-1}, \curve_i]| 
=  8h + 8\sum_{j = 1}^h s_j + 3|\bounded[\calG^{i-1}, \curve_i]|\ . 
\shortfull{$}{\]}

Since $h \leq \sqrt{|V_i|}$, we conclude that
\begin{align}
\cost[\calG^{i-1}, \curve_i] \leq 8\sqrt{|V_i|} + 8\sum_{j = 1}^h s_j + 3|\bounded[\calG^{i-1}, \curve_i]|\ .
\label{eqn:costbound}
\end{align}

Next, we upper bound the change in potential $\Phi(\calG^i) - \Phi(\calG^{i-1})$.  
Given a flarbable curve or sub-curve $\gamma$,
let $\calA(\gamma)$ denote the set of $\gamma$-faces. 
Recall that for a $\gamma$-face $f\in \calA(\gamma)$, $f'$ is the modified face of $f$.
Also, let $f_n$ be the new face created by $\calF^i$, i.e., the face of $\calG^i$ bounded by $\curve_i$.
Recall that for each face $f\in \bounded[\calG^{i-1}, \curve_i]$, $f$ is removed and the potential decreases by $\mu(f)\geq1$.
Using this, we can break up the summation to obtain the following:

\begin{align}
\Phi(\calG^i) - \Phi(\calG^{i-1}) 
&= \mu(f_n) +  \lambda \sum_{f \in \calA(\curve_i)} (\mu(f') - \mu(f)) - \lambda \sum_{f\in \bounded[\calG^{i-1}, \curve_i]} \mu(f)\\
&\leq \mu(f_n) +  \lambda \sum_{f \in \calA(\curve_i)} (\mu(f') - \mu(f)) - \lambda |\bounded[\calG^{i-1}, \curve_i]|\ .
\end{align}

\shortfull{}{\noindent} We now break up the first summation by independently considering the large faces in $\calL_i$ and the remaining smaller faces which are crossed by some flarbable sub-curve. Then

\begin{align}
\Phi(\calG^i) - \Phi(\calG^{i-1})\leq  &\ \mu(f_n)+  \lambda \sum_{j = 1}^h  \left( \sum_{f \in \calA(\gamma_j)}\left( \mu(f') -\mu(f) \right)\right) \\
&+ 
\lambda \sum_{f \in \calL_i}\left(\mu(f') - \mu(f)\right)- \lambda |\bounded[\calG^{i-1}, \curve_i]|.
\end{align}

Since each face can gain at most one edge, in particular we know that $\mu(f') - \mu(f) \leq 1$ for each $f\in \calL_i$. Moreover, $\mu(f_n) \leq \sqrt{|V_i|}$ by definition. Thus,
\shortfull{$}{\[}
\Phi(\calG^i) - \Phi(\calG^{i-1}) \leq \sqrt{|V_i|} + \lambda \sum_{j = 1}^h\left( \sum_{f \in \calA(\gamma_j)}\left( \mu(f') - \mu(f) \right)\right) + \lambda|\calL_i|- \lambda |\bounded[\calG^{i-1}, \curve_i]|.
\shortfull{$}{\]}

Note that $\mu(f) = |f|$ for each face $f\in \calA(\gamma_j)$, $1\leq j\leq h$.
Thus, applying Theorem~\ref{thm:delta} to the first summation, we get
\shortfull{$}{\[}
\Phi(\calG^i) - \Phi(\calG^{i-1}) \leq \sqrt{|V_i|} - \frac{\lambda}{2}\sum_{j = 1}^h s_j + \lambda|\calL_i| - \lambda |\bounded[\calG^{i-1}, \curve_i]|.
\shortfull{$}{\]}
Since there can be at most $\sqrt{|V_i|}$ faces of size $\sqrt{|V_i|}$, we know that $|\calL_i|\leq \sqrt{|V_i|}$. Therefore,

\begin{align}
\Phi(\calG^i) - \Phi(\calG^{i-1}) &\leq (\lambda+1) \sqrt{|V_i|} - \frac{\lambda}{2}\sum_{j = 1}^h s_j - \lambda |\bounded[\calG^{i-1}, \curve_i]|\label{eqn:potbound}
\end{align}

Putting~\eqref{eqn:costbound} and~\eqref{eqn:potbound} together, we get that 
\shortfull{$}{\[}
\cost[\calG^{i-1}, \curve_i] + \Phi(\calG^i) - \Phi(\calG^{i-1}) 
\leq (\lambda + 9)\sqrt{|V_i|} +  (8 - \frac{\lambda}{2}) \sum_{j = 1}^h s_j + (3 - \lambda)  |\bounded[\calG^{i-1}, \curve_i]|
\shortfull{$}{\]}

By letting $\lambda$ be a sufficiently large constant (namely $\lambda = 16$), we get that 
\shortfull{$}{\[}
\cost[\calG^{i-1}, \curve_i] + \Phi(\calG^i) - \Phi(\calG^{i-1}) = O(\sqrt{|V_i|}). \qedhere
\shortfull{$}{\]}

\end{proof}

\begin{corollary}\label{corollary:Flarbable sequence}
Let $G$ be a 3-regular plane graph with $\nu$ vertices. 
For a sequence $\scrC = \curve_1, \ldots, \curve_N$ of flarbable fleeqs for graph $G = (V, E)$ where $\nu = |V|$,
\shortfull{$}{\[}
\sum_{i = 1}^N \cost[\calG^{i-1}, \curve_i] = O(\nu + N\sqrt{\nu+ N})
\shortfull{$}{\]}
\end{corollary}

\begin{proof}
Using the result of Theorem~\ref{thm:amortizeoneflarb}, we can write
\shortfull{$}{\[}
\sum_{i = 1}^N  \cost[\calG^{i-1}, \curve_i] + \Phi(\calG^N) - \Phi(G) = O(N\sqrt{|V_i|}).
\shortfull{$}{\]}
Because $\Phi(G) = \lambda \sum_{f \in F}\mu(f)$, we know that $\Phi(G) = O(\nu)$. Analogously, since each flarb operation adds at most 2 vertices by Lemma~\ref{lemma:Net change in vertices after flarb}, we know that the number of vertices in $\calG^N$ is $O(\nu+ N)$ which, in turn, implies that $\Phi(\calG^N)  = O(\nu + N)$.
Therefore, 
\shortfull{$}{\[}
\sum_{i = 1}^N  \cost[\calG^{i-1}, \curve_i]   = O(N\sqrt{|V_i|} +  \Phi(G) - \Phi(\calG^N)) = O(\nu + N\sqrt{\nu+N}) \qedhere
\shortfull{$}{\]}
\end{proof}

\section{The lower bound}\label{sec:Lower Bound}
In Section~\ref{sec:Lower Bound}, we present an example of a 3-regular Halin graph $G$ with $\nu$ vertices---a tree with all leaves connected by a cycle to make it 3-regular---and a corresponding flarb operation with cost $\Omega(\sqrt{\nu})$ that yields a graph isomorphic to $G$. Because this sequence can be repeated, the amortized cost of a flarb is $\Theta(\sqrt{\nu})$.

Let $\nu = 2k(k+1) - 2$ for some positive integer $k$. The construction of the 3-regular graph with $\nu$ vertices is depicted in Figure~\ref{fig:LowerBound}. 
In this graph, we show the existence of a flarbable curve $\curve$ (dashed in the figure) such that the flarb operation on $G$ produces a graph $\calG(G, \curve)$ isomorphic to $G$. 
Moreover, $\curve$ crosses at least $k$ augmented $\curve$-faces and $k$ shrinking $\curve$-faces. 
Therefore, $\cost\geq k  = \Omega(\sqrt{\nu})$ by Lemma~\ref{lemma:The cost of a flarb}. 
Since the resulting graph is isomorphic to the original graph, this operation can be repeated in succession an arbitrarily high number times.
That is, there is a sequence of $N$ flarbable curves $\curve_1, \ldots, \curve_N$ such that $\sum_{i = 1}^N  \cost[\calG^{i-1}, \curve_i]  = \Omega (N\sqrt{\nu})$,

\begin{figure}[H]
\centering
\includegraphics{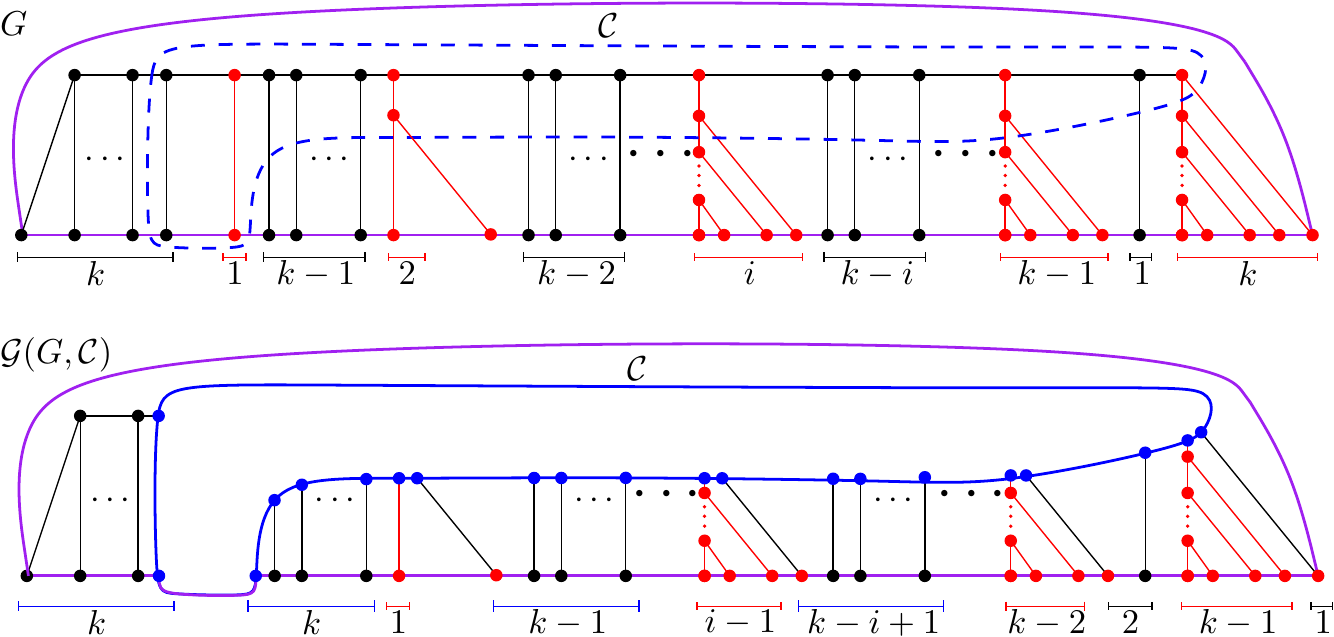}
\caption{\small A 3-regular graph $G$ with $\nu = 2k (k+1) - 2$ vertices. A flarbable curve $\curve$ induces a flarb such that $\calG(G, \curve)$ is isomorphic with $G$.}
\label{fig:LowerBound}
\end{figure}

\section{Computing the flarb}\label{sec:Computing the flarb}

In this section, we describe a data structure to maintain the Voronoi diagram of a set $S$ of $n$ sites in convex position as new sites are added to $S$. 
Our structure allows us to find the edges of each preserved face and ignore them, thereby focusing only on necessary modifications to the combinatorial structure. The time we spend in these operations is then proportional to the number of non-preserved edges. Since this number is proportional to the cost of the flarb, our data structure supports site insertions in time that is almost optimal (up to a polylogarithmic factor).

\subsection{Grappa trees}
Grappa trees~\cite{aronov2006data} are a modification of link-cut trees, a data structure introduced by~\citet{sleator1983data} to maintain the combinatorial structure of trees. They support the creation of new isolated vertices, the \emph{link} operation which adds an edge between two vertices in disjoint trees, and the \emph{cut} operation which removes an edge, splitting a tree into two trees.

We use this structure to maintain the combinatorial structure of the incrementally constructed Voronoi diagram $\V$ of a set $S$ of sites in convex position throughout construction. 
Recall that each insertion defines a flarbable curve $\curve$, namely the boundary of the Voronoi cell of the inserted site.
Our algorithm performs this flarb operation in time $O(\cost[\V, \curve] \log^7 n)$, where $n$ is the number of vertices inserted so far. 
That is, we obtain an algorithm whose running time depends on the minimum number of link and cut operations that the Voronoi diagram, which is a tree, must undergo after each insertion. 
Moreover, this Voronoi diagram answers nearest neighbor queries in $O(\log n)$ time. 

A grappa tree, as introduced by \citet{aronov2006data}, is a data structure is based on the worst-case version of the link-cut tree construction of \citet{sleator1983data}. 
This structure maintains a forest of fixed-topology trees subject to many operations, including {\sc Make-Tree, Link}, and {\sc Cut}, each in $O(\log n)$ worst-case time while using $O(n)$ space.

As in~\cite{aronov2006data,sleator1983data}, we decompose a rooted binary tree into a set of maximal vertex-disjoint downward paths, called \emph{heavy paths}, connected by tree edges called \emph{light edges}.
Each heavy path is in turn represented by a biased binary tree whose leaf-nodes correspond to the vertices of the heavy path. 
Non-leaf nodes represent edges of this heavy path, ordered in the biased tree according to their depth along the path.
Therefore, vertices that are higher (closer to the root) in the path correspond to leaves farther left in the biased tree.
Each leaf node $\ell$ of a biased tree $B$ represents an internal vertex $v$ of the tree which has a unique light edge $l_v$ adjacent to it. 
We keep a pointer from $\ell$ to this light edge. 
Note that the other endpoint of $l_v$ is the root of another heavy path which in turn is represented by another biased tree, say $B'$. 
We merge these two biased trees by adding a pointer from $\ell$ to the root of $B'$.
After merging all the biased trees in this way, we obtain the grappa tree of a tree $T$. 
A node of the grappa tree that is an internal vertex of its biased tree represents a heavy edge and has two children, whereas a node that is a leaf of its biased tree represents a vertex of the heavy path (and its unique adjacent light edge) and has only one child.
By a suitable choice of paths and biasing, as described in~\cite{sleator1983data}, the grappa tree has height $O(\log n)$.

In addition, grappa trees allow us to store left and right marks on each of its nodes, i.e., on each edge of $T$. 
To assign the mark of a node, grappa trees support the $O(\log n)$-time operation $\textsc{Left-Mark}(T,v,m_l)$ which sets the mark $m_l$ to every edge in the path from $v$ to the super root of $T$ ($\textsc{Right-Mark}(T,v,m_l)$ is defined analogously). 
In our setting, we use the marks of an edge $e$ to keep track of the faces adjacent to this edge in a geometric embedding of $T$. Since $T$ is rooted, we can differentiate between the left and the right faces adjacent to $e$.
\shortfull{A full description of the operations supported by grappa trees can be found in the appendix.}{}

The following definition formalizes the operations supported by a grappa tree.

\begin{definition}
 Grappa trees solve the following data-structural problem:
 maintain a forest of rooted binary trees
 with specified topology subject to:
 \begin{description}
 \item [\rm $T$ = Make-Tree$(v)$:] Create a new tree $T$ with a single
   internal vertex~$v$ (not previously in another tree).
 \item [\rm $T$ = Link$(v,w)$:]
   Given a vertex $v$ in one tree $T_v$ and the root $w$ of
   a different tree~$T_w$, connect $v$ and~$w$ and
   merge $T_v$ with $T_w$ into a new tree~$T$.
 \item [\rm $(T_1,T_2)$ = Cut$(e)$:] Delete the existing edge $e = (v,w)$
   in tree~$T$, splitting into $T$ two trees $T_1$ and $T_2$
   containing $v$ and $w$, respectively.
 \item [\rm Evert$(v)$:] Make external node $v$ the root of its tree,
   reversing the orientation (which endpoint is closer to the root)
   of every edge along the root-to-$v$ path.
 \item [\rm Left-Mark$(T,v,m_\ell)$:]
   Set the left mark of every edge on the root-to-$v$ path in $T$ to
   the new mark $m_\ell$, overwriting the previous left marks of these edges.
 \item [\rm Right-Mark$(T,v,m_r)$:]
   Set the right mark of every edge on the root-to-$v$ path in $T$ to
   the new mark $m_r$, overwriting the previous right marks of these edges.
 \item [\rm $(e,m^*_\ell,m^*_r)$ = Oracle-Search$(T,O_e)$:]
   Search for the edge $e$ in tree~$T$.
   The data structure can find $e$ only via \emph{oracle queries}:
   given two incident edges $f$ and $f'$ in~$T$, the provided oracle
   $O_e(f,f',m_\ell,m_r,m'_\ell,m'_r)$ determines in constant time
   which ``side'' of $f$ contains $e$, i.e., whether $e$ is in the component
   of $T - f$ that contains~$f'$, or in the rest of the tree
   (which includes $f$ itself).
   The data structure provides the oracle with the left mark $m_\ell$
   and the right mark $m_r$ of edge $f$, as well as the left mark $m'_\ell$
   and the right mark $m'_r$ of edge $f'$, and at the end, it returns
   the left mark $m^*_\ell$ and the right mark $m^*_r$ of the found edge~$e$.
 \end{description}
\end{definition}

\begin{theorem}\label{thm:grappa trees}
[Theorem 7 from~\cite{aronov2006data}]
A grappa tree maintains the combinatorial structure of a forest and supports each operation described above in $O(\log n)$ worst-case time per operation, where $n$ is the total size of the trees affected by the operation.
\end{theorem}

\subsection{The Voronoi diagram}

Let $S$ be a set of $n$ sites in convex position and let $\V$ be the binary tree representing the Voronoi diagram of $S$. We store $\V$ using a grappa tree. 
In addition, we assume that each edge of $\V$ has two \emph{face-markers}: its left and right markers which store the sites of $S$ whose Voronoi regions are adjacent to this edge on the left and right side, respectively.
While a grappa tree stores only the topological structure of $\V$, with the aid of the face-markers we can retrieve the geometric representation of $\V$. 
Namely, for each vertex $v$ of $\V$, we can look at its adjacent edges and their face-markers to retrieve the point in the plane representing the location of $v$ in the Voronoi diagram of $S$ in $O(1)$ time. 
Therefore, we refer to $v$ also as a point in the plane.
Recall that each vertex $v$ of $\V$ is the center of a circle that passes through at least three sites of $S$, we call these sites the \emph{definers} of $v$ and we call this circle the \emph{definer circle} of $v$.

\begin{observation}\label{obs:New site in circles}
Given a new site $q$ in the plane such that $S' = S\cup \{q\}$ is in convex position, the vertices of $\V$ that are closer to $q$ than to any other point of $S'$ are exactly the vertices whose definer circle encloses~$q$.
\end{observation}

Let $q$ be a new site such that $S' = S\cup \{q\}$ is in convex position.
Let $\rr{q}{S'}$ be the Voronoi region of~$q$ in the Voronoi diagram of~$S'$ and let $\rrb{q}{S'}$ denote its boundary.
Recall that we can think of $\V$ as a Halin graph by connecting all its leaves by a cycle to make it 3-regular. While we do not explicitly use this cycle, we need it to make our definitions consistent. 
In this Halin graph, the curve $\rrb{q}{S'}$ can be made into a closed curve by going around the leaf of $\V$ contained in $\rr{q}{S'}$, namely the point at infinity of the bisector between the two neighbors of $q$ along the convex hull of $S'$.
In this way, $\rrb{q}{S'}$ becomes a flarbable curve.
Therefore, we are interested in performing the flarb operation it induces which leads into a transformation of $\V$ into~$\mathcal V(S')$.

\subsection{Heavy paths in Voronoi diagrams}

Recall that for the grappa tree of $\V$, we computed a heavy path decomposition of $\V$.
In this section, we first identify the portion of each of these heavy paths that lies inside $\rr{q}{S'}$. 
Once this is done, we test if any edge adjacent to an endpoint of these paths is preserved. 
Then within each heavy path, we use the biased trees built on it to further find whether there are non-preserved edges on this heavy path. 
After identifying all the non-preserved edges, we remove them, which results in a split of $\V$ into a forest where each edge in $\rr{q}{S'}$ is preserved. 
Finally, we show how to link the disjoint components back to the tree resulting from the flarb operation. 

We first find the heavy paths of $\V$ whose roots lie in $\rr{q}{S'}$. 
Additionally, we find the portion of each of these heavy paths that lies inside $\rr{q}{S'}$. 

Recall that there is a leaf $\rho$ of $\V$ that lies in $\rr{q}{S'}$: the point at infinity of the bisector between the two neighbors of $q$ along the convex hull of $S'$. As a first step, we root $\V$ at $\rho$ by calling Evert$(\rho)$. In this way, $\rho$ becomes the root of $\V$ and all the heavy paths have a \emph{root} which is their endpoint closest to $\rho$.

Let $R$ be the set the of roots of all heavy paths of $\V$, and let $R_q = \{r\in R : r\in \rr{q}{S'}\}$.
We focus now on computing the set $R_q$. 
By Observation~\ref{obs:New site in circles}, each root in $R_q$ has a definer circle that contains~$q$. We use a dynamic data structure that stores the definer circles of the roots in $R$ and returns those circles containing a given query point efficiently.

\begin{lemma}\label{lemma:Chan data structure}
There is a fully dynamic $O(n)$-space data structure to store a set of circles (not necessarily with equal radii) that can answer queries of the form: Given a point $q$ in the plane, return a circle containing~$q$, where insertions take $O(\log^3 n)$ amortized time, deletions take $O(\log^6 n)$ amortized time, and queries take $O(\log^2 n)$ worst-case time.
\end{lemma}

\begin{proof}
\citet{chan2010dynamic} presented a fully dynamic randomized data structure that can answer queries about the convex hull of a set of $n$ points in three dimensions where insertions take $O(\log^3 n)$ amortized time, deletions take $O(\log^6 n)$ amortized time, and extreme-point queries take $O(\log^2 n)$ worst-case time. We use this structure to solve our problem, but first, we must transform our input into an instance that can be handled by this data structure.

Let $\mathscr C$ be the dynamic set of circles we want to store.
Consider the paraboloid-lifting which maps every point $(x,y) \to (x, y, x^2+ y^2)$. Using this lifting, we identify each circle $C\in \mathscr C$ with a plane $\pi_C$ in $\mathbb{R}^3$ whose intersection with the paraboloid projects down as $C$ in the $xy$-plane. 
Moreover, a point $q = (x,y)$ lies inside~$C$ if and only if point $(x,y, x^2+y^2)$ lies below the plane $\pi_C$. 

Let $\Pi = \{\pi_C : C\in \mathscr C\}$ be the set of planes corresponding to the circles in $\mathscr C$.
In the above setting, our query can be translated as follows:
Given a point $q' = (x,y, x^2 + y^2)$ on the paraboloid, find a plane $\pi_C\in \Pi$ that lies above $q'$. 

Using standard point-plane duality in $\mathbb{R}^3$, we can map the set of planes $\Pi$ to a point set $\Pi^*$, and a query point $q'$ to a plane $q^*$ such that a query translates to a \emph{plane query}:
Given a query plane $q^*$, find a point of $\Pi^*$ that lies below it. 

Using the data structure introduced by \citet{chan2010dynamic} to store $\Pi^*$, we can answer plane queries as follows. Consider the direction orthogonal to $q^*$ pointing in the direction below $q^*$. Then, find the extreme point of the convex hull of $\Pi^*$ in this direction in $O(\log^2 n)$ time. 
If this extreme point lies below $q^*$, return the circle of $\mathscr C$ corresponding to it.
Otherwise, we return that no point of $\Pi^*$ lies below $q^*$, which implies that no circle of $\mathscr C$ contains $q$.
Insertions take $O(\log^3 n )$ time while removals from the structure take $O(\log^6 n)$ time. 
\end{proof}

For our algorithm, we store each root in $R$ into the data structure given by Lemma~\ref{lemma:Chan data structure}.  Using this structure, we obtain the following result.

\begin{lemma}\label{lemma:Computing roots in polylog}
We can compute each root in $R_q$ in total $O(|R_q| \log^6 n)$ amortized time. 
\end{lemma}

\begin{proof}
After querying for a root whose definer circle contains $q$, we remove it from the data structure and query it again to find another root with the same property until no such root exists.
Since queries and removals take $O(\log^2 n)$ and $O(\log^6 n)$ time, respectively, we can find all roots in $R_q$ in $O(|R_q| \log^6 n)$ time.
\end{proof}

Given a root $r\in R$, let $h_r$ be the heavy path whose root is $r$.
Because the portion of $\V$ that lies inside $\rr{q}{S'}$ is a connected subtree, we know that, for each $r\in R_q$, the portion of the path $h_r$ contained in $\rr{q}{S'}$ is also connected.
In order to compute this connected subpath, we want to find the last vertex of $h_r$ that lies inside of $\rr{q}{S'}$, or equivalently, the unique edge of $h_r$ having exactly one endpoint in the interior of $\rr{q}{S'}$. We call such an edge the \emph{$q$-transition edge} of $h_r$ (or simply transition edge).

\begin{lemma}\label{lemma:Finding connected subpath in Voronoi}
For a root $r\in R_q$, we can compute the transition edge of $h_r$ in $O(\log n)$ time.
\end{lemma}
\begin{proof}
Let $e_r$ be the transition edge of $h_r$. 
We make use of the oracle search proper of a grappa-tree to find the edge $e_r$. To this end, we must provide the data structure with an oracle such that: given two incident edges $f$ and $f'$ in $\V$, the oracle determines in constant time which side of $f$ contains the edge $e_r$, i.e., whether $e_r$ is in the component of $\V \setminus f$ that contains $f'$, or in the rest of the tree (which includes $f$ itself).
The data structure provides the oracle with the left and the right marks of $f$ and $f'$.
Given such an oracle, a grappa tree allows us to find the edge $e_r$ in $O(\log n )$ time\shortfull{.}{ by Theorem~\ref{thm:grappa trees}.}

Given two adjacent edges $f$ and $f'$ of $\V$ that share a vertex $v$, we implement the oracle described above as follows.
Recall that the left and right face-marks of $f$ and $f'$ correspond to the sites of $S$ whose Voronoi region is incident to the edges $f$ and $f'$. 
Thus, we can determine the definers of the vertex $v$, find their circumcircle, and test whether $q$ lies inside it or not in constant time. Thus, by Observation~\ref{obs:New site in circles}, we can test in $O(1)$ time whether $v$ lies in $R_q$ or not and hence, decide if $e_r$ is in the component of $\V \setminus f$ that contains~$f'$, or in the rest of the tree.
\end{proof}

\subsection{Finding non-preserved edges.}

\begin{observation}\label{obs:Testing reincarnation}
Given a 3-regular graph $G$ and a flarbable curve $\curve$, if we can test whether a point is enclosed by $\curve$ in $O(1)$ time, then we can test whether an edge is preserved  in $O(1)$ time.
\end{observation}
\begin{proof}
First note that we can test  in $O(1)$ time whether an edge reappears by testing whether its two adjacent edges are fleeq-edges. Since a preserved edge is either an edge that reappears or a fleeq-edge adjacent to an edge that reappears, this takes only $O(1)$ time.
\end{proof}

Let $\Vq$ be the subtree induced by all the edges of $\V$ that intersect $\rr{q}{S'}$. 
Now, we work towards showing how to identify each non-preserved edge of $\Vq$ in the fleeq induced by $\rrb{q}{S'}$. 
For each root $r\in R_q$, we compute the transition edge $e_r$ of $h_r$ using Lemma~\ref{lemma:Finding connected subpath in Voronoi} in $O(\log n)$ time per edge.
Assume that $w$ is the vertex of $e_r$ that is closer to $r$ (or is equal to $r$). 
We consider each edge adjacent to $w$ and test whether it is preserved. 
Since each vertex of $\Vq$ has access to its definers via the face markers of its incident edges, we can test if this vertex lies in $\rr{q}{S'}$.
Thus, by Observation~\ref{obs:Testing reincarnation}, we can decide whether an edge of $\Vq$ is preserved in $O(1)$ time. 

We mark each non-preserved edge among them as \emph{shadow}. 
Because we can test whether an edge is preserved in $O(1)$ time, and since computing $e_r$ takes $O(\log n)$ time by Lemma~\ref{lemma:Finding connected subpath in Voronoi}, this can be done in total amortized $O(|R_q|\log n)$ time.
In addition, notice that if $h_r$ contains two adjacent vertices $u$ and $v$ such that the light edge of $u$ is a left edge while the light edge of $v$ is a right edge (or vice versa), then the edge $uv$ cannot be preserved; see Figure~\ref{fig:Bent Edge}.  
In this case, we say that $uv$ is a \emph{bent edge}. We want to mark all the bent edges in $\Vq$ as shadow, but first we need to identify them efficiently.

\begin{figure}[tb]
\centering
\includegraphics{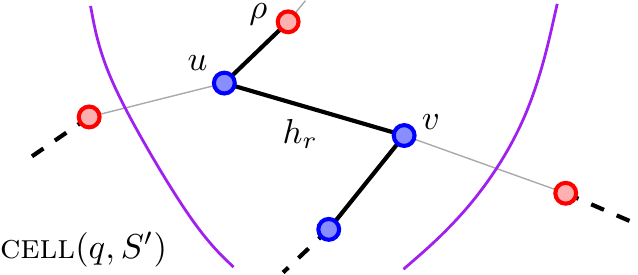}
\caption{\small Path $h_r$ contains two adjacent vertices $u$ and $v$ such that the light edge of $u$ is a left edge while the light edge of $v$ is a right edge. The edge $uv$ cannot be preserved.}
\label{fig:Bent Edge}
\end{figure}

Note that it suffices to find all the bent edges of $h_r$ for a given root $r\in R_q$, and then repeat this process for each root in $R_q$.
To find the bent edges in $h_r$, we further extend the grappa tree in such a way that the biased tree representing $h_r$ allows us to search for bent edges in $O(\log n)$ time. This extension is described as follows.
Recall that each leaf $s_v$ of a biased tree corresponds to a vertex $v$ of the heavy path and has a pointer to the unique light edge adjacent to $v$. Since each light edge is either left or right, we can extend the biased tree to allow us to search  in $O(\log n)$ time for the first two consecutive leaves where a change in direction occurs. From there, standard techniques allow us to find  the next change in direction in additional $O(\log n)$ time. Therefore, we can find all the bent edges of a heavy path $h_r$ in $O(\log n)$ time per bent edge.
After finding each bent edge in $h_r$, we mark it is as a shadow edge.

\begin{lemma}\label{lemma:Shadow is non-preserved}
An edge of $\Vq$ is a preserved edge if and only if it was not marked as a shadow edge.
\end{lemma}

\begin{proof}
Since we only mark non-preserved edges as shadow, 
we know that if an edge is preserved, then it is not shadow.

Assume that there is a non-preserved edge $uv$ of $\Vq$ that is not marked as shadow. 
If $uv$ is a heavy edge, then it belongs to some heavy path $h_r$ for some $r\in R_q$. 
We know that $uv$ cannot be the transition edge of $h_r$ since it would have been shadowed when we tested whether it was preserved. 
Thus, $uv$ is completely contained in $\rr{q}{S'}$. 
We can also assume that $uv$ is not a bent edge, otherwise $uv$ would have been shadowed. 
Therefore, the light children of $u$ and $v$ are either both left or both right children, say left.
Since $uv$ is not preserved, either the light child of $u$ or the light child of $v$ must be inside $\rr{q}{S'}$.
Otherwise if both edges cross the boundary of $\rr{q}{S'}$, then $uv$ is preserved by definition.

Assume that $u$ has a light left child $r'$ that is inside $\rr{q}{S'}$. That is, $r'$ must be the root of some heavy path and hence belongs to $R_q$.
However, in this case we would have checked all the edges adjacent to $u$ while processing the root $r'\in R_q$.
Therefore, every edge that is non-shadow and intersects $\rr{q}{S'}$ is a preserved edge.
\end{proof}

\begin{corollary}
It holds that $\sigma = \Theta(\cost[\V, \rrb{q}{S'}])$.
\end{corollary}

Let $\sigma$ be the number of shadow edges of $\V$, which is equal to the number of non-preserved edges by Lemma~\ref{lemma:Shadow is non-preserved}.
The following relates the size of $R_q$ with the value of~$\sigma$.

\begin{lemma}\label{lemma:Shadow bound roots}
It holds that $|R_q| = O(\sigma \log n)$. 
\end{lemma}

\begin{proof}
Given a root $r$ of $R_q$, let $p_r$ be the parent of $r$ and notice that the edge $r p_r$ is a light edge that is completely contained in $R_q$. Note that $p_r$ belongs to another heavy path $h_t$, for some $t\in R_q$. 
If $p_r$ is the endpoint of the transition edge of $h_t$ closest to the root, then we add a \emph{dependency pointer} from $r$ to $t$.
This produces a \emph{dependency graph} with vertex set $R_q$. Since there is only transition edge per heavy path, that the in-degree of each vertex in this dependency graph is one. Therefore, the dependency graph is a collection of (oriented)\emph{dependency paths}.

Since any path from a vertex to the root $\rho$ of $\V$ traverses $O(\log n)$ light edges, each dependency path has length $O(\log n)$.
Let $r\in R_q$ be the sink of a dependency path. 
Consider the light edge $r p_r$ and notice that it cannot be preserved, as $p_r$ is not incident to a transition edge. Therefore, we can charge this non-preserved edge to the dependency path with sink $r$. 
Since a non-preserved can be charged only once, we have that $\sigma$ is at least the number of dependency paths.
Finally, as each dependency path has length $O(\log n)$, there are at least $\Omega(|R_q|/\log n)$ of them.
Therefore $\sigma  = \Omega(|R_q| \log n)$, or equivalently, $|R_q| = O(\sigma \log n)$ which yields our result.
\end{proof}

\subsection{The compressed tree}

Let $\f$ be the forest obtained from $\Vq$ by removing all the shadow edges (this is just for analysis purposes, so far no cut has been performed).
Note that each connected component of $\f$ consists only of preserved edges that intersect $\rr{q}{S'}$. Thus, each component inside $\rr{q}{S'}$ is a comb, with a path as a \emph{spine} and each child of a spine vertex pointing to the same side; see Figure~\ref{fig:Compressing Comb}. 
Thus, we have right and left combs, depending on whether the children of the spine are left or right children. 

\begin{figure}[h!]
\centering
\includegraphics{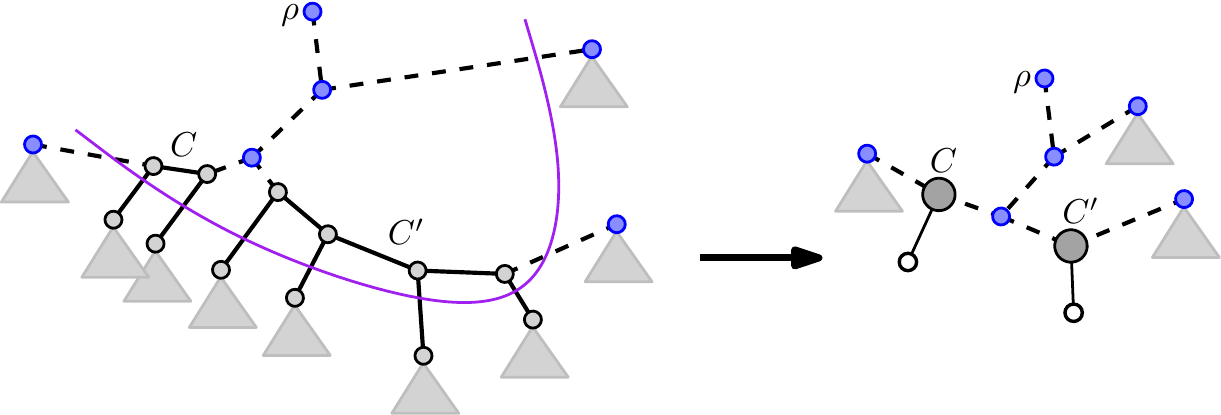}
\caption{\small Two combs of $\f$ that are compressed into super nodes with their respective dummy leaves. 
An Eulerian tour around the compressed tree provides us with the order in which the trees hanging outside of $\rr{q}{S'}$ should be attached. }
\label{fig:Compressing Comb}
\end{figure}

Our objective in the long term is to cut all the shadow edges and link the remaining components in the appropriate order to complete the flarb.
To this end, we would like to perform an Eulerian tour on the subtree $\Vq$ to find the order in which the subtrees of $\V\setminus \Vq$ that  hang from the leaves of $\Vq$ appear along this tour. 
However, this may be too expensive as we want to perform this in time proportional to the number of shadow edges and the size of $\Vq$ may be much larger.
To make this process efficient, we compress $\Vq$ by contracting each comb of $\f$ into a single super node.
By performing an Eulerian tour around this compressed tree, we obtain the order in which each component needs to be attached. We construct the compressed flarb and then we decompress as follows.

Note that each comb has exactly two shadow edges that connect it with the rest of the tree. 
Thus, we contract the entire component containing the comb into a single \emph{super node} and add a left or right dummy child to it depending on whether this comb was left or right, respectively; see Figure~\ref{fig:Compressing Comb}.
After the compression, the shadow edges together with the super nodes and the dummy vertices form a tree called the \emph{compressed tree} that has $O(\sigma)$ vertices and edges, where $\sigma$ is the total number of shadow edges. 

\begin{lemma}
We can obtain the compressed tree in $O(\sigma \log \sigma)$ time.
\end{lemma}
\begin{proof}
Notice that each shadow edge is adjacent to two faces---its left face and its right face.
Recall that each face ibounds the Voronoi cell of some site in $S$ and that each shadow edge has two markers pointing to the sites defining its adjacent faces.
Using hashing, we can group the shadow edges that are adjacent to the same face in $O(\sigma)$ time.
Since preserved faces have no shadow edges on their boundary, we have at most $O(\sigma)$ groups.

Finally, we can sort the shadow edges adjacent to a given face along its boundary. 
To this end, we use the convex hull position of the sites defining the faces on the other side of each of these shadow edges. Computing this convex hull takes $O(\sigma \log \sigma)$ time. 
Once the shadow edges are sorted along a face, we can walk and check whether consecutive shadow edges are adjacent. If they are not, then the path between them consists only of preserved edges forming a comb; see Figure~\ref{fig:Compressing Comb}. Therefore, we can compress this comb and continue walking along the shadow edges. 
Since each preserved edge that reappears is adjacent to a face containing at least one shadow edge (namely the face that is not preserved), all the combs will be compressed during this procedure.
\end{proof}

The compressed tree is then a binary tree where each super node has degree three and each edge is a shadow edge.
We now perform an Eulerian tour around this compressed tree and retrieve the order in which the leaves of this tree are visited. Some leaves are dummy leaves and some of them are original leaves of $\Vq$; see Figure~\ref{fig:Compressing Comb}.

\subsection{Completing the flarb}
We now proceed to remove each of the shadow edges which results in a (compressed) forest with $O(\sigma)$ components.
Note that each of the original leaves of $\Vq$ was connected with its parent via a shadow edge and hence it lies now as a single component in the resulting forest. 
For each of these original leaves of $\Vq$, we create a new \emph{anchor} node and link it as the parent of this leaf. 
Moreover, there could be internal vertices that become isolated. In particular this will be the case of the root $\rho$. 
These vertices are deleted and ignored for the rest of the process. 
\begin{figure}[h!]
\centering
\includegraphics{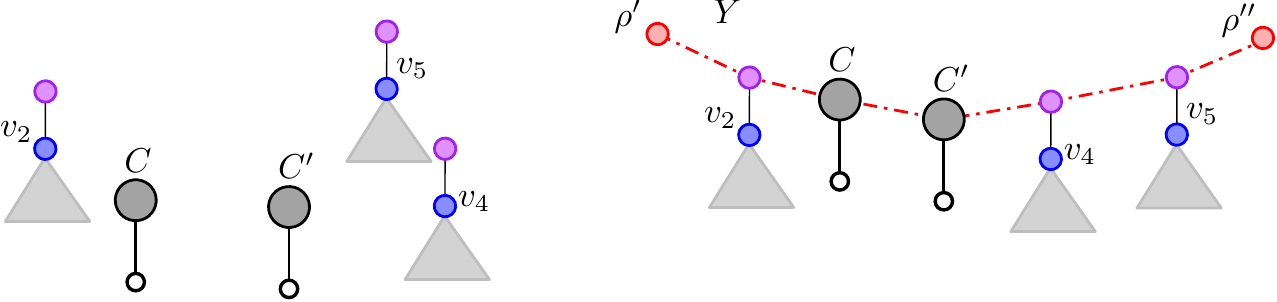}
\caption{\small Left: An anchor node is created for each isolated leaf of $\Vq$ and attached as its parent. Other isolated nodes are ignored. Right: A super comb is created connecting two new leaves $\rho'$ and $\rho''$ through a path. This path connects anchor and super nodes in the order retrieved by the Eulerian tour around the compressed tree.}
\label{fig:Linking after Cutting}
\end{figure}
To complete the flarb, we create two new nodes $\rho'$ and $\rho''$ which will be the two new leaves of the Voronoi diagram, one of them replacing $\rho$. Then, we construct a path with endpoints $\rho$ and $\rho'$ that connects the super nodes and the anchor nodes according to the traversal order of their leaves; see Figure~\ref{fig:Linking after Cutting}.
The resulting tree is a \emph{super comb} $Y$, where each vertex on the spine is either a super node or an anchor node, and all the leaves are either dummy leaves or original leaves of $\Vq$. 
Since we combined $O(\sigma)$ components into a tree, we need $O(\sigma)$ time.

We proceed now to decompress $Y$. To decompress a super node of $Y$ that corresponds to a comb, we consider the two neighbors of the super node in $Y$ and attach each of them to the ends of the spine of the comb. 
For an anchor node, we simply note that there is a component of $\V$ hanging from its leaf; see Figure~{\ref{fig:Decompression}}. 
In this way, we obtain all the edges that need to be linked. After the decompression, we end with the tree $\mathcal V(S')$ resulting from the flarb. 
Thus, the flarb operation of inserting $q$ can be implemented with $O(\sigma)$ link and cuts.

\begin{figure}[h!]
\centering
\includegraphics{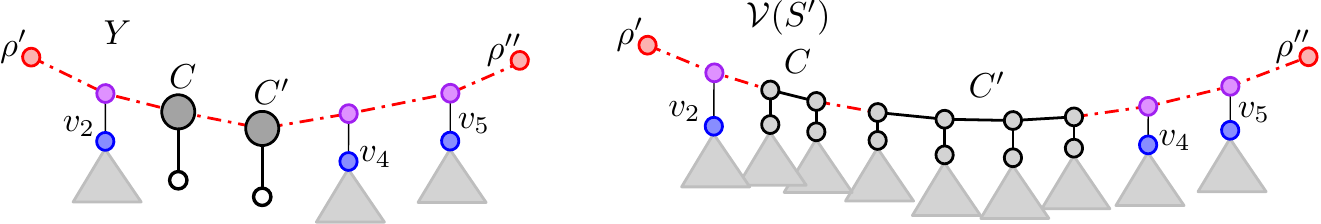}
\caption{\small The tree $\mathcal V(S')$ achieved after the decompression.}
\label{fig:Decompression}
\end{figure}

Recall that any optimal algorithm needs to perform a cut for each edge that is not preserved. 
Since each non-preserved edge is shadow by Lemma~\ref{lemma:Shadow is non-preserved}, the optimal algorithm needs to perform at least $\Omega(\sigma)$ operations. 
Therefore, our algorithm is optimal and computes the flarb using $\Theta(\sigma)$ link and cuts. 
Moreover, by Lemmas~\ref{lemma:Computing roots in polylog} and~\ref{lemma:Finding connected subpath in Voronoi} we can compute the flarb in $O(|R_q| \log^6 n + \sigma \log n)$ amortized time using $\Theta(\sigma)$ link and cuts. Since $|R_q| = O(\sigma \log n)$ by Lemma~\ref{lemma:Shadow bound roots}, we obtain the following.

\begin{theorem}\label{theorem:Cost of flarb}
The flarb operation of inserting $q$ can be implemented with $O(K)$ link and cuts, where $K$ is the cost of the flarb.
Moreover, it can be implemented in $O(K \log^7 n)$ amortized time.
\end{theorem}

\bibliographystyle{abbrvnat}
\bibliography{flarb}

\end{document}